\documentclass[letterpaper, 10 pt, conference]{ieeeconf}

\IEEEoverridecommandlockouts                              
\usepackage[utf8]{inputenc}
\usepackage[T1]{fontenc}

\let\labelindent\relax
\usepackage[english]{babel}
\usepackage{xcolor}
\usepackage{setspace}
\usepackage{graphicx}	
\usepackage{pgfplotstable} 
\usepackage{blkarray}
\usepackage[intlimits]{amsmath}
\usepackage{amsthm,thmtools}
\usepackage{mathtools} 
\usepackage{amssymb}
\usepackage[makeroom]{cancel}
\usepackage{hyperref}
\usepackage{stmaryrd}
\usepackage{mathrsfs}
\usepackage{array}
\usepackage[justification=centering]{caption}
\usepackage{enumitem}
\usepackage[bottom]{footmisc}

\newtheorem{theorem}{Theorem}
\newtheorem{lem}{Lemma}
\newtheorem{coro}{Corollary}
\newtheorem{definition}{Definition}

\DeclareMathOperator*{\argmin}{arg\,min}

\renewcommand{\Re}{\mathbb{R}}

\renewcommand{\paragraph}[1]{\smallskip\noindent\textbf{#1.} }

\author{Alexandre Kircher$^1$, Laurent Bako$^1$, Eric Blanco$^1$, Mohamed Benallouch$^2$ 
\thanks{$^1$ A. Kircher, L. Bako and E. Blanco are  with Universit\'{e} de Lyon, Laboratoire Amp\`{e}re (Ecole Centrale Lyon, CNRS UMR 5005), Ecully F-69134.
        {\tt\small E-mail: alexandre.kircher@ec-lyon.fr}}%
\thanks{$^2$M. Benallouch  is with ECAM Lyon,  40 Montée Saint-Barthélémy, 69321 Lyon, France.}
}
\title{Resilient State Estimation for Discrete-Time Linear Systems}

\begin{document}
\setstretch{.99}

\maketitle

\begin{abstract}
This paper proposes a resilient state estimator for LTI discrete-time systems.  The dynamic equation of the system is assumed to be affected by a bounded process noise. As to the available measurements, they are potentially corrupted by a noise of both dense and impulsive natures.  In this setting, we construct the estimator as the map which associates to the measurements, the minimizing set of an appropriate (convex) performance function. It is then shown that the proposed estimator enjoys the property of resilience, that is, it induces an estimation error which, under certain conditions, is independent of the extreme values of the (impulsive)  measurement noise. Therefore, the estimation error may be bounded while the measurement noise is virtually unbounded. Moreover, the expression of the bound depends explicitly on the degree of observability of the system being observed and on the considered performance function. Finally, a few simulation results are provided to illustrate the resilience property.
\end{abstract}



\section{Introduction}

We consider in this work the problem of designing state estimators which would be resilient against an (unknown) sparse noise sequence affecting the measurements. By sparse noise we refer here to a  signal sequence which is of impulsive nature, that is, a sequence which is most of the time equal to zero,  except at a few instants where it can take on arbitrarily large values. 
The problem is relevant for example, in the supervision of  Cyber-Physical Systems~\cite{cardenas_secure_2008}. In this application, the supervisory data may be collected by spatially distributed sensors and then sent to a distant processing unit through some communication network. During the transmission, the data may incur intermittent packet losses or adversarial attacks consisting in e.g., the injection of arbitrary signals. 

 This estimation problem was investigated through many different approaches. 
 Since the measurements are assumed to be affected by a sequence of outliers which is sparse in time, a natural scheme of solution to the state estimation problem may be to first detect the  occurrences of the nonzero instances of that sparse noise, remove the corrupted data and then proceed with classical estimation methods such as the Kalman filter or Luenberger type of observer \cite{mishra_secure_2017,pasqualetti_attack_2013}. 
Another category of approaches, which are inspired by some recent results in compressive sampling \cite{Candes08-SPM,Foucart13-Book}, rely on sparsity-inducing optimization techniques. A striking feature of these methods  is that they do not treat separately the tasks of detection, data cleaning and estimation. Instead, an implicit discrimination of the wrong data is induced by some  specific properties of the to-be-minimized cost function.  One of the first works that puts forward this approach for the resilient state estimation problem is the one reported in \cite{fawzi_secure_2014}. There, it is assumed that only a fixed number of sensors are subject to attacks (sparse but otherwise arbitrary disturbances). The challenge then resides in the fact that at each time instant, one does not know which sensor is compromised. Note however that the assumptions in \cite{fawzi_secure_2014} were quite restrictive as no process noise or measurement noise (other than the sparse attack signal) was considered. 
These limitations open ways for later extensions in many directions. For example, \cite{shoukry_event-triggered_2016} suggests  a reformulation which  reduces computational cost  by using the concept of event-triggered update;   
\cite{pajic_attack-resilient_2017} considers an observation model which includes dense noise along with the sparse attack signal. In \cite{chang_secure_2018},  the assumption of a fixed number of attacked sensors is relaxed. Finally,  the recent paper \cite{Han19-TAC} proposes a unified framework for analyzing resilience capabilities of most of these optimization-based estimators. Although a bound on the estimation error was derived in this paper, it is not quantitatively related to the properties (e.g., observability) of the dynamic system being observed.

The contribution of the current paper is the design of a (convex)  optimization-based resilient estimator for LTI discrete-time systems.  The available model of the system assumes bounded noise in both the dynamics and the observation equation with the latter being possibly affected by an unknown but sparse attack signal. Contrary to the settings in some existing works, we  did not impose here any  restriction  on the number of sensors which are subject to attacks, that is, any sensor can be compromised at any time. Our main theoretical result concerns the resilience analysis of the proposed estimator. 
 We show that the estimation error associated with the new estimator can be made, under certain conditions, insensitive to the amplitude of the attack signal. Our bound, although necessarily conservative, has the important advantage of being explicitly expressible in function of the properties of the considered dynamic system. This makes it a valuable qualitative tool for assessing the impact of the estimator's design parameters and that of the system matrices on the quality of the estimation.  
For example, it reflects the intuition that the more observable the system is, the larger the number of instances of gross values (of the output noise) it can handle and the smaller the error bound.

\paragraph{Outline} The rest of the paper is organized as follows. The estimation setting is defined in Section \ref{part:ii}. In Section \ref{part:iii} we elaborate on the proposed optimization-based estimator: Necessary technical tools are introduced in Section \ref{subsec:Preliminaries} for the statement and the proof of the main result in Section \ref{subsec:Resilience}. Section  \ref{part:v} illustrates the performance of the estimation method in simulation; Section \ref{part:vi} provides some concluding remarks.

\paragraph{Notations} Throughout this paper, $\Re_{\geq 0}$ (respectively $\Re_{>0}$) designates the set of nonnegative (respectively positive) reals. 
We note $\Re^a$ the set of (column) vectors with $a$ real elements and for any vector $z$ in $\Re^a$, $z_i$ with $i$ in $\{1,...,a\}$ is the $i$-th component of $z$. 
Moreover, $\Re^{a\times b}$ is the set of real matrices with $a$ rows and $b$ columns. If $M\in \Re^{a\times b}$, then $M^{\top}$ will designate the transposed matrix of $M$. Notation $\lVert \cdot \rVert$ will represent a given norm over a given set (which will  be specified when necessary). $\lVert \cdot \rVert_2$ is the Euclidean norm, defined by $\lVert z\rVert_2=\sqrt{z^\top z}$ for all $z$ in $\Re^a$. $\left\|\cdot\right\|_1$ will designate the $\ell_1$-norm, defined by $\left\|z\right\|_1=\sum_{i=1}^a |z_i|$ for $z\in\Re^a$. For a finite set $\mathcal{S}$, the notation $|\mathcal{S}|$ will  refer to the cardinality of $\mathcal{S}$.  

\section{The estimation  Problem}\label{part:ii}

Consider the following discrete-time  Linear Time-Invariant (LTI) system 
\begin{equation}\label{eq:sys}
\Sigma : \left\{\begin{array}{r  l}
x_{t+1} &= Ax_t+w_t \\
y_t &=Cx_t+f_t  
\end{array}
\right.
\end{equation}
where $x_t\in\Re^n$ is the state vector at time $t$, $y_t\in\Re^{n_y}$ is the output vector at time $t$; $A\in\Re^{n\times n}$ the dynamic matrix of the system and $C\in\Re^{n_y\times n_y}$ is the observation matrix. 
$w_t\in\Re^n$  and $f_t\in\Re^{n_y}$ model respectively the process noise and the output noise both of which are unknown.\\ 
We shall however make the informal assumptions that $\left\{w_t\right\}$  is bounded with a relatively small amplitude. 
As to the sequence  $\left\{f_t\right\}$, it  can take on potentially arbitrarily large values, that is, no explicit bound is imposed on its amplitude. This type of noise can model for example, ordinary measurement noise (of `moderate amount') together with  intermittent faulty measurements, attack signals or packet losses on data transmitted over a communication network. For convenience, one can also view $f_t$ as the sum of two noise components, a dense noise, representing a bounded noise induced by the sensors,   and a \textit{sparse noise} sequence,  \textit{i.e.}, a noise whose instances are equal to zero most of the time but whose non-zero elements can take on arbitrary values.

\paragraph{Problem}The problem considered in this paper is one of estimating the states $x_0,\ldots,x_{T-1}$ of the system \eqref{eq:sys} on a time period $\mathcal{T}$ given $T$ measurements $y_0,...,y_{T-1}$ of the system output. 
We shall seek  an accurate estimate of the state despite the uncertainties in the system equations \eqref{eq:sys} modeled by $w_t$ and $f_t$ the characteristics of which are described above. In particular, we would like the  to-be-designed  estimator to produce an estimate such that the estimation error is, when possible, independent of the maximum amplitude of $\left\{f_t\right\}$. Such an estimator will then be called resilient.

\section{Resilient optimization-based estimator}\label{part:iii}

We propose a convex optimization-based solution to the state estimation problem defined above. Given the system matrices $A$ and $C$ and  $T$ output measurements $y_0,...,y_{T-1}$, consider a performance function 
$F:\Re^{n\times T}\rightarrow \Re_{\geq 0}$ defined by 
\begin{equation}\label{eq:F(Z)}
	F(Z)= \lambda\sum_{t\in\mathcal{T}'}\lVert z_{t+1}-Az_t\rVert^2_2+\sum_{t\in\mathcal{T}}\left\|y_t-Cz_t\right\|_1, 
\end{equation}
where $\mathcal{T}=\left\{0,\ldots,T-1\right\}$,  $\mathcal{T}'=\left\{0,\ldots,T-2\right\}$ and   $Z=\big(\begin{matrix}z_0 & \cdots & z_{T-1}\end{matrix}\big)$, i.e., the vectors $z_t\in \Re^n$ are the columns of the matrix $Z$. Here, $\lambda>0$ is a user-defined parameter which aims at balancing the contributions of the two terms involved in the expression of the performance index $F$. This idea of weighting the terms contained in $F$ could also be done differently depending on the time index, for example by taking terms of the form $\lVert W_t(z_{t+1}-Az_t)\rVert_2^2$ and $\lVert V_t(y_t-Cz_t)\rVert_1$, where $W_t$ and $V_t$ would be positive-definite weighting matrices. 

Let $\mathcal{P}(\Re^{n\times T})$ denote the collection of subsets of $\Re^{n\times T}$. Then the proposed estimator  is defined as the set-valued map $\Psi:\Re^{n_y\times T}\rightarrow \mathcal{P}(\Re^{n\times T})$ which maps the  available measurements $Y\triangleq\big(\begin{matrix}y_0 & \cdots & y_{T-1}\end{matrix}\big)$ to the subset $\Psi(Y)$ of $\Re^{n\times T}$ defined by 
\begin{equation}\label{eq:Psi(Y)}
	\Psi(Y)=\argmin_{Z\in \Re^{n\times T}} F(Z). 
\end{equation}
By assuming that the pair $(A,C)$ is observable, it can be checked that $F$ is coercive, \textit{i.e.}, it satisfies $\lim_{\left\|Z\right\|\rightarrow +\infty} F(Z)=+\infty$ for any norm $\left\|\cdot\right\|$ on $\Re^{n\times T}$. It follows that the estimator $\Psi$ expressed in \eqref{eq:Psi(Y)} is well-defined in the sense that the underlying optimization problem in  \eqref{eq:Psi(Y)} admits a solution \cite{rockafellar_convex_1970}. Note however that the minimizer need not be unique. 
Moreover, since the objective function $F$ is convex, the elements of the so-defined state estimator  $\Psi(Y)$ can be  determined efficiently for a given $Y$. Many numerical solvers can be used for this purpose, see e.g. \cite{grant_cvx_2017,aps_mosek_2001,sturm_using_1999} for the computational aspects.   


The rest of the paper will focus on assessing the resilience properties of the estimator \eqref{eq:Psi(Y)}. For this purpose we need some preliminary technical results.  

\subsection{Preliminaries}\label{subsec:Preliminaries}

To begin with the analysis,  we introduce some useful technical tools, the first of which is the  class of $\mathcal{K}_\infty$ functions (see, e.g., \cite{Kellett14}). This class of functions will be used to measure the increasing rate of the estimation error.
\begin{definition}[class-$\mathcal{K_\infty}$ functions]
A function $\alpha:\Re_{\geq 0}\rightarrow\Re_{\geq 0}$ is said to be of class-$\mathcal{K_\infty}$ if it is continuous, zero at zero, strictly increasing and satisfies $\lim_{s\rightarrow +\infty }\alpha(s)=+\infty$.
\end{definition}
Using this definition we can state a technical lemma which will play an important role in the analysis. 
\begin{lem}\label{eq:lem:minimum-value}
Let $G:\Re^{n\times m}\rightarrow\Re_{\geq 0}$ be a nonnegative continuous function satisfying the following properties:
\begin{itemize}
	\item Positive definiteness: $G(S)=0$ if and only if $S=0$
	\item Relaxed homogeneity: There exists a $\mathcal{K}_\infty$ function $\sigma$ such that $G(S)\geq \sigma(\frac{1}{\lambda})G(\lambda S)$ for all $\lambda\in \Re_{>0}$.
\end{itemize}
Then for any norm $\left\|\cdot\right\|$ on $\Re^{n\times m}$,  there exists $d>0$  such that for all $S\in\Re^{n\times m}$, $G(S)\geq d \sigma(\left\|S\right\|)$. 
\end{lem}
\begin{proof}
We start by observing that the unit hypersphere $\mathcal{D}=\left\{S\in \Re^{n\times m}: \left\|S\right\|=1\right\}$ is a compact set in the topology induced by the norm $\left\|\cdot\right\|$.  By the extreme value theorem, $G$ being continuous, admits necessarily a minimum value  on $\mathcal{D}$, i.e., there is $S^\star\in \mathcal{D}$ such that $G(S)\geq d\triangleq G(S^\star)>0$ for all $S\in \mathcal{D}$. For any nonzero $S\in \Re^{n\times m}$, $\dfrac{S}{\left\|S\right\|}\in \mathcal{D}$ so that  $G(\dfrac{S}{\left\|S\right\|})\geq d$. On the other hand, by the relaxed homogeneity of $G$, 
$$G(S)\geq \sigma(\left\|S\right\|) G(\dfrac{S}{\left\|S\right\|})\geq d \sigma(\left\|S\right\|).$$ Moreover, this inequality holds for $S=0$. It therefore holds true for any $S\in \Re^{n\times m}$. 
\end{proof}

For future uses in the paper, consider now the function $H:\Re^{n\times T}\rightarrow\Re_{\geq 0}$ defined by
\begin{equation}\label{eq:Fprime}
	H(Z)= \dfrac{\lambda}{2}\sum_{t\in\mathcal{T}'}\lVert z_{t+1}-Az_t\rVert^2_2+\sum_{t\in\mathcal{T}}\left\|Cz_t\right\|_1 \end{equation}
	
Note the resemblance between $F(Z)$ and $H(Z)$. They only differ by the absence of $y_t$ in the second term of $H$ and the factor of the first term which is $\lambda$ in the first case and $\lambda/2$ in the second.	
	
\begin{lem}[Lower Bound on $H$]\label{lem:lb}
Let  $\lVert\cdot\rVert$ be a norm on $\Re^{n\times T}$.  Consider the function $H$ defined in \eqref{eq:Fprime} under the assumption that $(A,C)$ is observable. Then 
\begin{equation}
H(Z)\geq Dq(\lVert Z\rVert) \quad \forall Z\in \Re^{n\times T}
\end{equation}
where $q:\Re_{\geq0}\rightarrow\Re_{\geq 0}$ is the function defined by 
\begin{equation}\label{eq:def_q}
\forall \alpha\in\Re_{\geq0},\:q(\alpha)=\min(\alpha,\alpha^2)
\end{equation}
and
\begin{equation}\label{eq:def_d}
D=\min_{\lVert Z \rVert=1}H(Z)>0.
\end{equation}
\end{lem}
\begin{proof}
The idea of the proof  is to check that $H$ satisfies the conditions  of Lemma \ref{eq:lem:minimum-value} and then apply it to conclude. First, note that  continuity and nonnegativity  of $H$ are obvious. As to the relaxed homogeneity property, it can be checked straightforwardly that it holds with $\sigma=q$. Finally, setting $H(Z)=0$ implies that $z_{t+1}=Az_t$ and $Cz_t=0$ for all $t=0,\ldots,T-1$. It immediately  follows that $CA^tz_0=0$ and so, $\mathcal{O}z_0=0$ where $\mathcal{O}=\begin{pmatrix} C^{\top} & \cdots & (CA^{n-1})^{\top}\end{pmatrix}^{\top}$ is the observability matrix of the system. By the observability assumption, we get that $z_0=0$ and consequently, that $Z=0$. Therefore $H$ is positive-definite. The statement of the lemma now follows by applying Lemma \ref{eq:lem:minimum-value}. 
\end{proof}

\noindent To proceed further, let us introduce a few notations. We use the notation $\mathcal{I}=\left\{1,\ldots,n_y\right\}$ to denote a label set for the sensors described by the observation equation in \eqref{eq:sys} and $\mathcal{T}=\left\{0,\ldots,T-1\right\}$ to the set of time indexes. For $i\in \mathcal{I}$, $c_i^\top$ denotes the $i$-th row of the observation matrix $C$. \\
The next definition introduces a parameter to gauge the resilience properties of an estimator of the form defined in \eqref{eq:Psi(Y)}. 
\begin{definition}[$r$-Resilience index $p_r$]
Let   $r$ be a nonnegative integer. Assume that the system  $\Sigma$ in \eqref{eq:sys} is  observable.  We define the $r$-Resilience index of the estimator $\Psi$ in \eqref{eq:Psi(Y)} (when applied to $\Sigma$) as the real number $p_r$ given by 
\begin{equation}\label{eq:def_pr}
p_r=\sup_{\substack{Z\neq 0\\ Z\in\Re^{n\times T}}}\sup_{\substack{\Lambda_r\subset \mathcal{I}\times \mathcal{T}\\|\Lambda_r|=r}}\dfrac{\sum_{(i,t)\in \Lambda_r} \left|c_i^\top z_t\right|}{H(Z)}
\end{equation}
where $H$ is as defined in~\eqref{eq:Fprime}. The supremum is taken here over all nonzero $Z$ in $\Re^{n\times T}$ and over all subsets  $\Lambda_r$ of $\mathcal{I}\times \mathcal{T}$ with cardinality $r$. 
\end{definition}
\noindent The index $p_r$ can be interpreted as a quantitative measure of the observability of the system $\Sigma$. The observability is needed here to ensure that the denominator $H(Z)$ of \eqref{eq:def_pr} is different from zero whenever $Z\neq 0$ (see the positive definiteness proof of $H$ in the proof Lemma \ref{lem:lb} above).  Furthermore, it should be remarked that $\sum_{(i,t)\in \Lambda_r} \left|c_i^\top z_t\right|\leq H(Z)$ for any $\Lambda_r\subset \mathcal{I}\times \mathcal{T}$,  which implies that the defining suprema of  $p_r$ are well-defined. 

What the $r$-Resilience parameter assesses is how much the estimator can handle data corruption as it represents the worst ratio between the weight of $r$ corrupted estimates (which take any value and be potentially placed anywhere in time) and the weight of the whole estimated trajectory. As a result, the lower $p_r$ is, the more resilient the estimator is expected to be. The next section gives more background to the introduction of $p_r$ and which role it exactly plays in the resilience analysis of the estimator.

From a computational viewpoint  we observe that the parameter $p_r$ is hard to compute in general. In effect, obtaining $p_r$ numerically would require solving a nonconvex and combinatorial optimization problem. This is indeed a common characteristic of the concepts which are usually used to assess resilience; for example the popular  Restricted Isometry Property (RIP) constant ~\cite{candes_restricted_2008} is comparatively as hard to evaluate. Nevertheless, if we restrict attention to  estimation problems where the process noise $\left\{w_t\right\}$ would be identically equal to zero, then by adding in \eqref{eq:def_pr} the additional constraint that $z_{t+1}=Az_t$, $p_r$ can be exactly computed  using the method in \cite{sharon_minimum_2009} or more cheaply overestimated using the one in \cite{bako_class_2017}.

\subsection{Characterization of the resilience property}\label{subsec:Resilience}

The main result of this paper  consists in the characterization of the resilience property of the state estimator \eqref{eq:Psi(Y)}. 
More specifically, our result states that  the estimation error, \textit{i.e.}, the difference between the real state trajectory and the estimated one, is upper bounded by a bound  which does not depend on the amplitude of the outliers contained in $\left\{f_t\right\}$ provided that the number of such outliers is below some threshold. 

Before stating the main theorem, let us introduce a last notation to be used in the analysis. 
Let $\varepsilon\geq 0$ be a given number. For any admissible  sequence $\left\{f_t\right\}_{t\in \mathcal{T}}$ in \eqref{eq:sys}, we can  split the index set $\mathcal{I}\times \mathcal{T}$ into two disjoint label sets,  
\begin{equation}\label{eq:def_teps}
	\mathcal{J}_\varepsilon=\left\{(i,t)\in \mathcal{I}\times \mathcal{T}: |f_{it}|\leq \varepsilon\right\},
\end{equation}
indexing those\footnote{$f_{it}$ denotes the $i$-th entry of the vector $f_t$.} $f_{it}$ which are bounded by $\varepsilon$ and  $\mathcal{J}_\varepsilon^c=\left\{(i,t)\in \mathcal{I}\times \mathcal{T}: |f_{it}|> \varepsilon\right\}$ indexing those $f_{it}$ which are possibly unbounded. It is important to keep in mind that $\varepsilon$ is just a parameter for decomposing the noise sequence in two parts in view of the analysis (and not a bound on $f_{it}$). 
 The particular situation where $\varepsilon=0$ reflects the approach where one would view any nonzero $f_{it}$ as an outlier.
\begin{theorem}[Upper bound on the estimation error]\label{th:lb}
Consider the system $\Sigma$ defined by~\eqref{eq:sys} with output measurement $Y$ and consider the estimator \eqref{eq:Psi(Y)}. Let $\varepsilon\in\Re_{\geq 0}$   and $r=|\mathcal{J}_\varepsilon^c|$. If $\Sigma$ is observable and $p_r<1/2$, then for all $\hat{X}=\begin{pmatrix}
\hat{x}_0 &\cdots &\hat{x}_{T-1}
\end{pmatrix}\in \Psi(Y)$, 
\begin{equation}\label{eq:th_ub}
\lVert E\rVert \leq h\left(\dfrac{2\beta_\Sigma(\varepsilon)}{D(1-2p_r)}\right)
\end{equation} 
where $E=\begin{pmatrix}
\hat{x}_0-x_0 &\cdots &\hat{x}_{T-1}-x_{T-1}
\end{pmatrix}$, $\lVert \cdot \rVert$ is any given norm on $\Re^{n\times T}$, $\beta_\Sigma(\varepsilon)$ is defined by
\begin{equation}
\beta_\Sigma(\varepsilon)=\lambda\sum_{t\in\mathcal{T}'}\lVert w_t\rVert_2^2+\sum_{(i,t)\in \mathcal{J}_\varepsilon}|f_{it}|,
\end{equation}
 the function $h:\Re_{\geq 0}\rightarrow\Re_{\geq 0}$ is defined by 
\begin{equation}\label{eq:def_h}
\forall \alpha\in\Re_{\geq 0},\: h(\alpha)=\max\big(\alpha,\sqrt{\alpha}\big)
\end{equation}
and $D$ is given as in~\eqref{eq:def_d} from the norm $\lVert \cdot \rVert$.
\end{theorem}
\begin{proof} By definition~\eqref{eq:Psi(Y)} of the  estimator $\Psi$, it holds that for all $\hat{X}\in \Psi(Y)$,  $F(\hat{X})\leq F(X)$, that is, 
\begin{equation}
	\label{eq:deb_proof}
	\begin{aligned}
	\lambda\sum_{t\in\mathcal{T}'}&\lVert \hat{x}_{t+1}-A\hat{x}_t\rVert^2_2+\sum_{t\in\mathcal{T}} \left\|y_t-C\hat{x}_t\right\|_1\\ 
	&\leq \lambda\sum_{t\in\mathcal{T}'}\lVert x_{t+1}-Ax_t\rVert^2_2+\sum_{t\in\mathcal{T}}\left\|y_t-Cx_t\right\|_1\\
	&= \lambda\sum_{t\in\mathcal{T}'}\lVert w_t\rVert^2_2+\sum_{t\in\mathcal{T}}\left\|f_t\right\|_1. 
	\end{aligned}
\end{equation}
Next, we derive a lower bound on the left hand side of \eqref{eq:deb_proof}.
For every $t$ in $\mathcal{T}$, let $e_t=\hat{x}_t-{x}_t$. Then
\begin{equation}\label{eq:INEQ2}
	\begin{aligned}
	\lVert  \hat{x}_{t+1}-A\hat{x}_{t}\rVert_2^2& =\lVert \hat{x}_{t+1}-x_{t+1}-A(\hat{x}_t-{x}_t)+w_t\rVert_2^2   \\
	& \geq \lVert e_{t+1}-Ae_t+w_t\rVert_2^2\\
	&\geq\dfrac{1}{2}\lVert e_{t+1}-Ae_t\rVert_2^2-\lVert w_t\rVert_2^2. 
	\end{aligned} 
\end{equation}
The last inequality uses the identity (see Lemma~\ref{lem:n2} in Appendix~\ref{app:a} for a proof)
\begin{equation}\label{eq:tri_gen}
\lVert z_1-z_2\rVert_2^2\geq \dfrac{1}{2}\lVert z_1 \rVert_2^2-\lVert z_2\rVert_2^2 \: \: \forall (z_1,z_2) \in \Re^n\times \Re^n.
\end{equation}
 Similarly, we can write
\begin{align*}
\left\|y_t-C\hat{x}_t\right\|_1&=\left\|y_t-Cx_t -C(\hat{x}_t-x_t)\right\|_1\\
&=\left\|f_t+Ce_t\right\|_1
\end{align*}
As a consequence, the second term of the left-hand-side of \eqref{eq:deb_proof} is expressible as 
$$\sum_{t\in\mathcal{T}} \left\|y_t-C\hat{x}_t\right\|_1=\sum_{(i,t)\in \mathcal{I}\times \mathcal{T}}\left|f_{it}+c_i^\top e_t\right|. $$

Now, depending on if the couple $(i,t)$ belongs to $\mathcal{J}_\varepsilon$ or not, we apply the triangle inequality property of the absolute value differently, the two cases being
$$
\begin{aligned}
\forall (i,t)\in \mathcal{J}_\varepsilon,  \quad & \left|f_{it}+c_i^\top e_t\right|\geq  |c_i^\top e_t|-|f_{it}|\\
\forall (i,t)\in \mathcal{J}_\varepsilon^c,\quad  &\left|f_{it}+c_i^\top e_t\right|\geq |f_{it}|-|c_i^\top e_t|
\end{aligned}
$$
It follows that 
$$
\begin{aligned}
	\sum_{t\in\mathcal{T}} \left\|y_t-C\hat{x}_t\right\|_1 \geq& \sum_{(i,t)\in \mathcal{J}_\varepsilon} |c_i^\top e_t| -\sum_{(i,t)\in \mathcal{J}_\varepsilon^c} |c_i^\top e_t|\\
	& -\sum_{(i,t)\in \mathcal{J}_\varepsilon} |f_{it}| +\sum_{(i,t)\in \mathcal{J}_\varepsilon^c} |f_{it}|. 
\end{aligned}
$$
Combining this with \eqref{eq:deb_proof} and \eqref{eq:INEQ2} and re-arranging, yields
\begin{multline}\label{eq:INEQ3}
\dfrac{\lambda}{2}\sum_{t\in\mathcal{T}'} \lVert e_{t+1}-Ae_t\rVert_2^2+\sum_{(i,t)\in \mathcal{J}_\varepsilon}|c_i^\top e_t|-\sum_{(i,t)\in \mathcal{J}_\varepsilon^c}|c_i^\top e_t| \\
\leq 2\Big(\lambda \sum_{t\in\mathcal{T}'}\lVert w_t\rVert_2^2+\sum_{(i,t)\in \mathcal{J}_\varepsilon}|f_{it}|\Big)
\end{multline}
On the right hand side of \eqref{eq:INEQ3}, we recognize $2\beta_\Sigma(\varepsilon)$ as in~\eqref{eq:th_ub}. As to the term on the left hand side, it is equal to $H(E)-2\sum_{(i,t)\in\mathcal{J}_\varepsilon^c}|c_i^\top e_t|$. 

Independently, $|\mathcal{J}_\varepsilon^c|=r$ so by definition~\eqref{eq:def_pr} of the index $p_r$,
\begin{equation}\label{eq:app_pr}
\sum_{(i,t)\in\mathcal{J}_\varepsilon^c}|c_i^\top e_t|\leq p_rH(E)
\end{equation}
Consequently, it follows from \eqref{eq:INEQ3} and \eqref{eq:app_pr} that 
$$(1-2p_r)H(E)\leq H(E)-2\sum_{(i,t)\in\mathcal{J}_\varepsilon^c}|c_i^\top e_t|\leq  2\beta_{\Sigma}(\varepsilon). $$
Since $p_r$ is assumed to be smaller than $1/2$,   $1-2p_r>0$. Therefore, we can write
\begin{equation}\label{eq:proof_i}
 H(E)\leq \dfrac{2\beta_\Sigma(\varepsilon)}{1-2p_r}
\end{equation}
Thanks to  Lemma~\ref{lem:lb}, we have $H(E)\geq Dq(\lVert E\rVert)$ for any given norm $\lVert\cdot\rVert$ on $\Re^{n\times T}$. This implies that 
$$ q(\lVert E\rVert)\leq \dfrac{2\beta_\Sigma(\varepsilon)}{D(1-2p_r)}$$
Now observe that the function $h$  defined in~\eqref{eq:def_h}, is the inverse function of $q$, meaning that for every $\lambda\in\Re_{\geq 0}$, $h(q(\lambda))=\lambda$. Moreover, $h$ is an increasing function. Applying $h$ to both members of the previous inequality gives the desired result.
\end{proof}

The resilience property of the estimator \eqref{eq:Psi(Y)} lies here in the fact that, under the conditions of Theorem \ref{th:lb}, the bound in \eqref{eq:th_ub} on the estimation error does not depend on the magnitudes of the extreme values of the noise sequence $\left\{f_{it}\right\}_{(i,t)\in \mathcal{I}\times \mathcal{T}}$. Considering in particular the function $\beta_\Sigma(\varepsilon)$, we remark that it can be overestimated as follows  
\begin{equation}
\beta_\Sigma(\varepsilon)\leq \lambda\sum_{t\in\mathcal{T}'}\lVert w_t\rVert_2^2+|\mathcal{J}_\varepsilon|\varepsilon. 
\end{equation}
We recognize two terms in the upper bound of $\beta_\Sigma(\varepsilon)$: 
(i) the first one is a sum which simply represents the uncertainty brought by the dense noise $w_t$ over the whole state trajectory and which does not depend on $\varepsilon$; 
(ii) the second one is a bound on the sum of those instances of $f_{it}$ whose magnitude is smaller that $\varepsilon$. \\ 
Because $\beta_\Sigma$ is a function of $\varepsilon$, the bound in \eqref{eq:th_ub}  represents indeed a family of bounds  parameterized by $\varepsilon$. Since $\varepsilon$ is a mere analysis device, a question would be how to select it for the analysis to achieve the smallest bound. Such values, say $\varepsilon^\star$, satisfy  
$$\varepsilon^\star\in \argmin_{\varepsilon\geq 0} \left\{h\Big(\dfrac{2\beta_\Sigma(\varepsilon)}{D(1-2p_r)}\Big):  \: r=|\mathcal{J}_\varepsilon^c|, \: p_r<1/2\right\}.  $$

Another interesting point is that the inequality stated by Theorem~\ref{th:lb} holds for any norm on $\Re^{n\times T}$. Note though that the value of the bound depends (through the parameter $D$ defined in \eqref{eq:def_d}) on the specific norm used to measure the estimation error.  Moreover, different choices of the performance-measuring norm will result in different geometric forms for the uncertain set, that is, the ball (in the chosen norm) centered at the true state with radius equal to the upper bound displayed in \eqref{eq:th_ub}.

We also observe that  the smaller the parameter $p_r$ is, the tighter the error bound will be, which suggests that the estimator is more resilient when $p_r$ is lower. A similar reasoning applies to the number $D$ which is desired to be large here.  
These two parameters (i.e., $p_r$ and $D$) reflect properties of the system whose state is being estimated. They can be interpreted, to some extent, as measures of the degree of observability of the system. In conclusion, the estimator inherits partially its resilience property from characteristics of the system being observed. This is consistent with the well-known fact that the more observable a system is, the more robustly its state can be estimated from output measurements.

Finally, an interesting property of the estimator can be stated in the absence of dense noise:

\begin{coro}
Consider the system $\Sigma$  defined by~\eqref{eq:sys} and let $r=|\mathcal{J}_0^c|$ (which means that we consider every nonzero occurrence of $f_{it}$ as an outlier). If $p_r<1/2$, and if  $w_t=0$ for all $t$, then the estimator defined by~\eqref{eq:Psi(Y)} retrieves exactly the state trajectory of the system.
\end{coro}

\begin{proof}
This follows directly from the fact that $\beta_\Sigma(0)=0$ in the case where there is no dense noise $w_t$ and $\varepsilon=0$. 
\end{proof}
Therefore, we have the exact recoverability of every state of the system \eqref{eq:sys} by the estimator when there is no process noise. According to our analysis, the number of outliers that can be handled by the estimator in this case can be underestimated by 
\begin{equation}
	\max\big\{r: p_r<1/2\big\}. 
\end{equation}

\section{Simulation Results}\label{part:v}

In this section, we present the simulation results of a system desgined as~\eqref{eq:sys} with
\begin{equation*}
	A=\begin{pmatrix}
	-0.11	&-0.34\\
	-0.34	&0.46
	\end{pmatrix},\:C=\begin{pmatrix}
	1.4 & -0.94
	\end{pmatrix}
\end{equation*}
$w_t$ is a gaussian white noise of unit variance. The dense component of $f_t$, which will be called $v_t$ in this section, is a gaussian white noise of signal-to-noise ratio equal to 30dB, while the sparse component of $f_t$, which will be called $s_t$, is a sparse vector whose non-zero elements are randomly selected and given a random value: as a result of this structure, we note $y_{wt}=Cx_t+v_t$ the uncorrupted output of the system. The estimated states were then obtained by directly solving the optimisation problem defined in~\eqref{eq:Psi(Y)} with $\lambda=1/5$ through CVX \cite{grant_cvx_2017}. To give a basis for comparison, we also estimated the state of the system through a Rauch-Tung-Striebel smoother which is an extension of the Kalman filter to offline estimation~\cite{gelb_applied_1974}.

\begin{figure}[t]
\vspace{0.5em}
\centering
\hspace{-0.1em}
\includegraphics[width=0.45\textwidth]{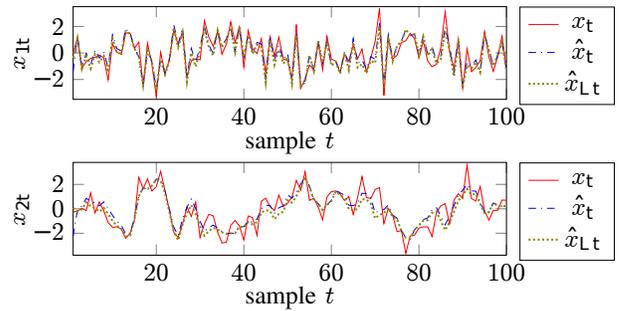}
\caption{State of the system and its estimates (resilient estimator and smoother) in absence of sparse noise}
\label{fig:x_sb}
\end{figure}

\begin{figure}[t]
\centering
\hspace{-0.1em}
\includegraphics[width=0.45\textwidth]{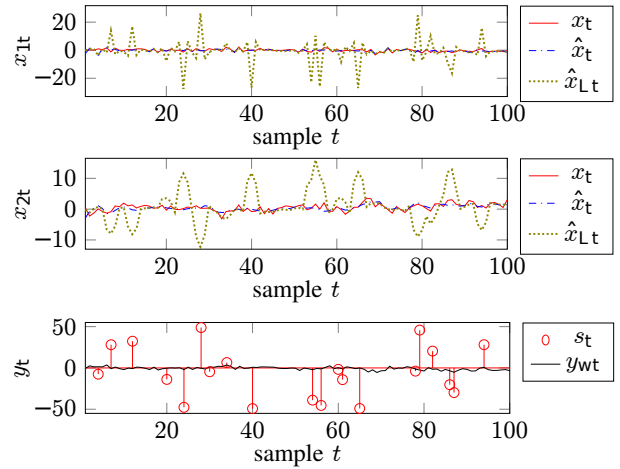}
\caption{State, estimated states (through resilient estimation and smoothing) and output of the system in presence of sparse noise}
\label{fig:xy_abl}
\end{figure}
\unskip

Figure~\ref{fig:x_sb} presents the classic case where there is no sparse noise corrupting the output of the system. This is the scenario handled by classic estimators such as the Kalman Filter or in our case the Rauch-Tung-Striebel smoother. We can however notice that our estimator gives satisfying results, fitting the trajectory of the real state and giving very similar results to the smoother. It is all the more interesting as our estimator does not take into account the statistical properties of the  noises involved in the system, contrary to the smoother which requires a tuning to approach the variance of those noises.

Figure~\ref{fig:xy_abl} now presents the case where twenty corrupted values were added to the output of the system. The smoother tries to compensate the attacks, as it can be noted that the estimate diverges when a corruption occurs, but it is entirely normal given that the Kalman filter theory is designed around noises in the form of white gaussian processes only. Figure~\ref{fig:xy_ab} compares the trajectory of the real state and the estimated state obtained through our resilient estimator. Even in the presence of corrupted measurements of arbitrarily large magnitude, the estimator still manages to efficiently track the trajectory of the real states, showing that its performance are not really degraded in that case.

\unskip
\begin{figure}[h]
\centering
\includegraphics[width=0.45\textwidth]{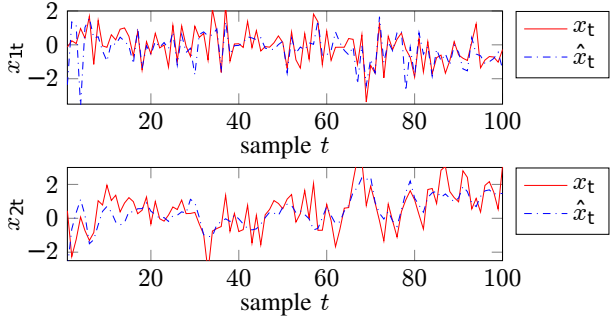}
\caption{State of the system and its estimate (resilient estimator) in presence of sparse noise}
\label{fig:xy_ab}
\end{figure}
\unskip

\section{Conclusion}\label{part:vi}
In this paper, we considered the problem of estimating the state of linear discrete-time systems in the face of uncertainties modeled as process  and measurement noise in the system equations. The measurement noise sequence assumes values of possibly arbitrarily large amplitude which occur intermittently in time.   
For this problem we  proposed an estimator based on the resolution of a convex optimization problem. In particular, we proved a resilience property for the proposed estimator,  that is, the resulting estimation error is bounded by a bound which is independent of the extreme values of the measurement noise provided that the number of occurrences (over time and over the whole set of sensors) of such extreme values is limited. 
Future works will aim at  generalizing the resilient properties to a wider class of estimators and applying the estimation framework to relevant practical cases.

\appendix

\subsection{Additional elements to the proof of Theorem 1}\label{app:a} 

\begin{lem}\label{lem:n2}
Let $G:\Re^{n\times m}\rightarrow\Re_{\geq 0}$ be a convex function satisfying the properties of positive definiteness and relaxed homogeneity (for a given $\mathcal{K}_{\infty}$ function $\sigma$) as both defined in Lemma~\ref{eq:lem:minimum-value}. Then, for all $(S_1,S_2)\in\Re^{n\times m}\times\Re^{n\times m}$,
\begin{equation}
G(S_1-S_2)\geq 2 \sigma(1/2)G(S_1)-G(S_2)
\end{equation}
\end{lem}

\begin{proof}
As $G$ is convex,
\begin{equation}
G\left(\dfrac{1}{2}(S_1-S_2)+\dfrac{S_2}{2}\right)\leq \dfrac{1}{2}G(S_1-S_2)+\dfrac{1}{2}G(S_2)
\end{equation}
which, by multiplying the whole inequality by 2, can be rewritten as
\begin{equation}\label{eq:lem3_pro}
G(S_1-S_2)\geq 2G(S_1/2)-G(S_2)
\end{equation}
Moreover, by assumption, $G$ verifies the relaxed homogeneity property with a $\mathcal{K}_\infty$ function $\sigma$: it entails that
\begin{equation}
\forall S_1\in\Re^{n\times m},\: G(S_1/2)\geq \sigma(1/2)G(S_1)
\end{equation}
which, when injected in~\eqref{eq:lem3_pro}, gives the desired result.
\end{proof}

In the case where $G=\lVert \cdot \rVert_2^2$, as norms are homogeneous, for every $\lambda\in\Re_{>0}$ and $z\in\Re^n$, $G(z)=G(\lambda z)/\lambda^2$. It follows that Lemma~\ref{lem:n2} can be applied to $G$ for $\sigma$ such that $\forall \alpha\in\Re_{\geq 0}$, $\sigma(\alpha)=\alpha^2$, yielding
\begin{equation}
\lVert z_1-z_2\rVert_2^2\geq \dfrac{1}{2}\lVert z_1 \rVert_2^2-\lVert z_2\rVert_2^2 \: \: \forall (z_1,z_2) \in \Re^n\times \Re^n.
\end{equation}

\bibliographystyle{abbrv}

\end{document}